\documentclass{article}
\usepackage{authblk}
\usepackage{amssymb,amsmath,amsfonts}
\usepackage{amsthm}
\usepackage{newtxtext,newtxmath}
\usepackage{latexsym}
\usepackage{verbatim}
\usepackage{enumerate}
\usepackage{graphicx}
\usepackage{url}
\usepackage[ruled,noend]{algorithm2e}
\usepackage{color}
\SetKw{KwOr}{or}%
\SetKw{KwAnd}{and}%
\SetKw{KwContinue}{continue}%
\SetKwFunction{localSearch}{DetLocalSearch}
\SetKwFunction{randLocalSearch}{RandLocalSearch}
\SetKwFunction{DetNRC}{DetNRC}
\SetKwFunction{RandNRC}{RandNRC}

\SetKwProg{Fn}{}{:}{}

\newtheorem{theorem}{Theorem}
\newtheorem{lemma}{Lemma}

\theoremstyle{remark}

\theoremstyle{definition}

\begin{document}
%\title{Exact Algorithms for No-Rainbow Coloring and Phylogenetic Decisiveness}
%\author{Ghazaleh Parvini\thanks{College of Information & Computer Sciences,  University of Massachusetts Amherst, Amherst, MA 01003, {\tt gparvini@cs.umass.edu}}}
%\author{David Fern\'{a}ndez-Baca\thanks{Department of Computer Science, Iowa State University, Ames, IA 50011, USA, {\tt fernande@iastate.edu}.}
%}
\title{Exact Algorithms for No-Rainbow Coloring and Phylogenetic Decisiveness}
\author[1]{Ghazaleh Parvini}
\author[2]{David Fern\'{a}ndez-Baca}
\affil[1]{College of Information \& Computer Sciences,  University of Massachusetts Amherst, Amherst, MA 01003, USA}
\affil[2]{Department of Computer Science, Iowa State University, Ames, IA 50011, USA}
\date{\empty}

\maketitle

\begin{abstract}
The input to the no-rainbow hypergraph coloring problem is a hypergraph $H$ where every hyperedge has $r$ nodes. The question is whether there exists an $r$-coloring of the nodes of $H$ such that all $r$ colors are used and there is no rainbow hyperedge --- i.e., no hyperedge uses all $r$ colors.  The no-rainbow hypergraph $r$-coloring problem is known to be NP-complete for $r \geq 3$.  The special case of $r=4$ is the complement of the phylogenetic decisiveness problem. 
%We showed an exact algorithm for this problem in \cite{ParviniBraughtFBISBRA2020} which solved no-rainbow r coloring for $r=3$ and $r=4$, in $O^*(1.89^n)$ and $O^*(2.81)^n$ time, respectively.
Here we present a deterministic algorithm that solves the no-rainbow $r$-coloring problem in  $O^*((r-1)^{(r-1)n/r})$ time and a randomized algorithm that solves the problem in $O^*((\frac{r}{2})^n)$ time.
\end{abstract}

\section{Introduction}

The hypergraph no-rainbow $r$-coloring problem asks whether there is an assignment to one of $r$ colors to each of the nodes in a hypergraph such that all $r$ colors are used, but no edge has $r$ nodes with distinct colors.  The problem is perhaps the most well known in the class of ``surjective constraint satisfaction problems'' \cite{bodirsky2012complexity}. It was recently shown to be NP-complete \cite{ZhukNoRainbow2020}, even for $r=3$.  

The case $r =4$ is of interest in building phylogenetic trees from incomplete data.  References  \cite{SteelSanderson2010,SandersonTerracesScience2011} formulate the question of determining whether such data suffices to reconstruct a unique tree as a combinatorial problem, termed ``decisiveness.''  The decisiveness problem is equivalent to the complement of the no-rainbow $4$-coloring problem \cite{ParviniBraughtFBISBRA2020}.  
 
In \cite{ParviniBraughtFBISBRA2020} we presented exact algorithms for no-rainbow $r$-coloring for $r =3$ and $r =4$, that run in time $O^*(1.89^n)$ and $O^*(2.81^n)$ time, respectively\footnote{The $O^*$-notation is a variant of $O$-notation that ignores polynomial factors \cite{FominKratsch2010}.}. Here we present more efficient randomized and deterministic algorithms based on local search. 
We note that local search has been applied earlier to obtain exact algorithms for the $k$-SAT problem \cite{Schoening1999,MoserScheder2011} (see also \cite{SchoeningToran2013}).

\section{Preliminaries \label{sec:prelims}}

We write $[r]$ to denote the set $\{1,2, \dots, r\}$, for some positive integer $r$.

%\section{No-Rainbow $r$-Colorings of Hypergraphs}
A \emph{hypergraph} 
$H$ is a pair  $H = (X,E)$, where $X$ is a set of elements called \emph{nodes} or \emph{vertices}, and 
$E$ is a set of non-empty subsets of $X$ called \emph{hyperedges} or \emph{edges} \cite{BergeHypergraphs1984}.  Two nodes $u, v \in V$ are \emph{neighbors} if $\{u,v\} \subseteq e$, for some $e \in E$.   A hypergraph $H = (X,E)$ is \emph{$r$-uniform}, for some integer $r > 0$, if each hyperedge of $H$ contains exactly $r$ nodes. 

%A \emph{chain} in a hypergraph $H = (X,E)$ is an %alternating sequence $v_1, e_1, v_2, \dots ,$
%$e_s, v_{s+1}$ of nodes and edges of $H$ such that:
%(1) $v_1, \dots , v_s$ are all distinct nodes of $H$,
%(2) $e_1, \dots, e_s$ are all distinct edges of $H$, and
%(3) $\{v_j,v_{j+1}\} \in e_j$ for $j \in \{1, \dots %,s\}$.
%A \emph{path is a chain where the first and last nodes, %$v_1$ and $v_{s+1}$ are distinct.
%Two nodes $u, v \in X$ are \emph{connected} in $H$, denoted $u \equiv v$,  if there exists a chain in $H$ that starts at $u$ and ends at $v$. The relation $u \equiv v$ is an equivalence relation \cite{BergeGraphsHypergraphs1973}; the equivalence classes of this relation are called the \emph{connected components} of $H$.  $H$ is \emph{connected} if it has only one connected component; otherwise $H$ is \emph{disconnected}.

%\subsection{No-Rainbow Colorings and Decisiveness}  
Let $H = (X,E)$ be a hypergraph and $r$ be a positive integer. An \emph{$r$-coloring} of $H$ is a mapping $c : X \rightarrow [r]$.  For node $v \in X$, $c(v)$ is the \emph{color} of $v$.  Throughout this paper, $r$-colorings are assumed to be \emph{surjective}; that is, for each $i \in [r]$, there is at least one node $v \in X$ such that $c(v) = i$. The \emph{Hamming distance} between two colorings $c$ and $c'$ of $H$, denoted $d(c,c')$, equals the number of nodes $v \in X$ such that $c(v) \neq c'(v)$.  
Two colorings $c$ and $c'$ of $H$ \emph{agree} on set $F \subseteq X$ if $c(v) = c'(v)$ for all $v \in F$. 

Let $c$ be an $r$-coloring of $H$.  Edge $e \in E$ is a \emph{rainbow edge with respect to $c$} if, for each $i \in [r]$, there is at least one $v \in e$ such that $c(v) = i$.  Coloring $c$ \emph{induces a rainbow edge} if there is an edge $e\in E$ such that  $e$ is a rainbow edge with respect to $c$.
% Thus, if $H$ is an $r$-uniform rainbow $r$-colored hypergraph, then, for some edge $e \in E$, all nodes in $e$ have distinct colors.
A \emph{no-rainbow $r$-coloring} of $H$ is a surjective $r$-coloring $c$ of $H$ such that $c$ induces no rainbow edge.
Given an $r$-uniform hypergraph $H = (X,E)$, the \emph{no-rainbow $r$-coloring problem}  asks whether $H$ has a no-rainbow $r$-coloring \cite{bodirsky2012complexity}. 

%The \emph{Hamming ball with radius $g$} around a coloring $c$, denoted $Ham_{g} (c)$, is the set of all colorings within Hamming distance $g$ from $c$; i.e.
%$$Ham_{g} (c) = \{\tilde{c}| d(\tilde{c},c) \leq g\}$$

\section{Deterministic Local Search}

Our deterministic algorithm for the no-rainbow $r$-coloring  (Algorithm \ref{alg:generalnonRainbowlocalSearchFreeze}) consists of a procedure \DetNRC that conducts multiple searches for a no-rainbow coloring, starting from different colorings, by invoking a procedure \localSearch.  \DetNRC takes as its argument an $r$-uniform hypergraph $H = (X,E)$ and returns $1$ (true) if $H$ has a no-rainbow $r$-coloring, and 0 (false) otherwise\footnote{For simplicity, our algorithm only returns true or false.  The algorithm can be easily modified to return a no-rainbow coloring, if one exists.}.
To explain our algorithm in detail, we need to introduce some concepts.

\begin{algorithm}%[H]

\Fn(){\DetNRC{$H$}}{
%\KwIn{$r$-uniform hypergraph $H = (X,E)$.}
%\KwOut{}
\SetAlgoLined
\SetNoFillComment
\DontPrintSemicolon
\ForEach{initial candidate pair $(c,F)$}{
	\lIf{\localSearch{$H, (c,F),\frac{(r-1)n}{r}$}}{\Return $1$}
}
\Return $0$
}

\medskip

\Fn(){\localSearch{$H,(c,F),g$}}{
%\KwIn{$r$-uniform hypergraph $H = (X,E)$, a surjective coloring $c$ of $H$, a set $F \subseteq X$ of nodes with fixed color, and an integer $g \ge 0$.}
%\KwOut{$1$ if $c$ can be transformed into a no-rainbow $r$-coloring of $H$ by changing the colors of at most $g$ nodes in $X \setminus F$; 0 otherwise.}
\SetAlgoLined
\SetNoFillComment
\DontPrintSemicolon
\lIf{$g = 0$ \KwAnd $c$ induces a rainbow edge in $H$}{
	\Return $0$}
\lIf{$c$ induces a rainbow edge $e$ such that $e \subseteq F$}{
	\Return $0$}
\lIf{$c$ is a no-rainbow coloring}{
	\Return $1$}
\lIf{$|e \cap F| \neq r-1$ for every $e \in E$}{
%\lIf{$|e \cap F| \le r-2$ for every rainbow edge $e$}{
%	\ForEach{v \in F}{$c(v) = 1$ to all $v \in X$}
	\Return $1$}
Choose any rainbow hyperedge $e \in E$ such that $|e \cap F| = r-1$\;
Let $v$ be the node of $e$ such that $v \not\in F$ \;
\ForEach{$j \in [r] \setminus \{c(v)\}$}{
	Let $c'$ be the coloring of $H$ where $c'(v) = j$ and $c'(u) = c(u)$ for all $u \in X \setminus \{v\}$ \;
	\lIf{$\localSearch(H, (c', F \cup \{v\}),g-1)$}{\Return $1$}
}
\Return $0$
}
\caption{A deterministic algorithm for no-rainbow $r$-coloring.}
\label{alg:generalnonRainbowlocalSearchFreeze}
\end{algorithm}

A \emph{candidate pair} for $H = (X,E)$ is a pair $(c,F)$ where $c$ is an $r$-coloring of $H$ and  $F$ is a subset of $X$ such that for each $i \in [r]$ there is a node $v \in F$ such that $c(v) = i$ (thus, $c$ is surjective and $|F| \ge r$).  Once a node is added to set $F$, its color is not allowed to change. We refer to the nodes in $F$ as \emph{frozen} nodes.

Starting from an initial candidate pair, \localSearch performs a series of steps, where each step takes us from the current candidate pair $(c,F)$ to a new candidate pair $(c',F')$, where $c'$ is obtained by changing the color of some node $v \not\in F$ to one of the other $r-1$ colors and where $F' = F \cup \{v\}$.  A step may lead to a dead end --- a candidate pair $(c,F)$ where $c$ induces a rainbow edge $e$ and every node in $e$ is frozen.  In this case, we backtrack and try another candidate pair. 

To limit the running time of \localSearch, we put a bound on the  \emph{search radius} to be explored; that is, on the maximum Hamming distance allowed between the initial coloring and any coloring encountered during the search. 
To guarantee correctness, we must choose the starting candidate pairs and the search radius appropriately.

An \emph{initial candidate pair} for $H = (X,E)$ is a candidate pair $(c,F)$ where $|F| = r$ and every node in $X \setminus F$ has the same color (which, obviously, must be the same as that of one of the nodes in $F$.  Note that the number of initial candidate pairs is ${n \choose r} \cdot r$, which is polynomial for fixed $r$.
%The next lemma shows that, if a no-rainbow coloring $c^*$ exists, there is an initial candidate pair $(c,F)$ that is within distance $\frac{(r-1)n}{r}$ of $c^*$

\begin{lemma}\label{lem:initial}
Let $H = (X,E)$ be an $r$-uniform hypergraph with $n$ nodes.  Suppose $H$ has a no-rainbow $r$-coloring $c^*$.  Then, there exists an initial candidate pair $(c,F)$ such that $c$ and $c^*$ agree on $F$ and $d(c, c^*) \leq \frac{(r-1)n}{r}$. 
\end{lemma}

\begin{proof}
Let $(c,F)$ be the initial candidate pair constructed as follows.  Set $F$ consists of any $r$ nodes in $X$ that are assigned distinct colors by $c^*$.  For each $v \in F$, let $c(v) = c^*(v)$. For each $v \in X \setminus F$, let $c(v) = i$, where color $i \in [r]$ is chosen arbitrarily among the colors used at least $\frac{n}{r}$ times by $c^*$.  Then, $(c,F)$ is an initial candidate pair with $d(c,c^*) \leq (r-1)n/r$.
\end{proof}

The previous lemma puts a limit on the search radius that must be explored, provided we start with the right candidate pair.  In fact, during the search, we may find candidate pairs $(c,F)$ where $c$ induces a rainbow edge, but yet we can immediately conclude that a rainbow coloring exists.  The next lemma characterizes this situation.

\begin{lemma}\label{lem:freeze1}
%Let $H = (X,E)$ be an $r$-uniform hypergraph.  Suppose that there exists a candidate pair $(c,F)$ for $H$ such that for each edge $e \in E$, (i) if $e$ is a rainbow edge with respect $c$, then $|e\cap F| \leq r-2$, and (ii) if $e$ is not a rainbow edge with respect to $c$, then $|e\cap F| \neq r-1$.  Then, $H$ admits a no-rainbow coloring $c^*$ that agrees with $c$ on $F$.
Let $H = (X,E)$ be an $r$-uniform hypergraph.  Suppose that $(c,F)$ is a candidate pair for $H$ such that there does not exist a rainbow edge $e$ where $e \subseteq F$.  If $|e \cap F| \neq r-1$ for each edge $e \in E$, then $H$ admits a no-rainbow $r$-coloring that agrees with $c$ on $F$.
\end{lemma}

\begin{proof}
Let $c^*$ be the $r$-coloring of $H$ where $c^*(v) = c(v)$ for all $v \in F$ and $c^*(v) = 1$ for all $v \in X \setminus F$. Then $c^*$ is a no-rainbow $r$-coloring of $H$ because (a) $c^*$ is surjective, (b) no edge $e \in E$ such that $|e\cap F| = r$ is a rainbow edge, and (c) for each $e \in E$ such that $|e\cap F| \leq r-2$, $c$ assigns the same color to at least two nodes in $e$.
\end{proof}

Let $(c,F)$ be a candidate pair for a hypergraph $H = (X,E)$.
We say that $(c,F)$ is \emph{within $g$ of a no-rainbow coloring}, for some integer $g \ge 0$, if there exists a candidate pair $(c',F')$ where $F \subseteq F'$ such that $c$ and $c'$ agree on $F$, $d(c,c^*) \le g$, and either $c'$ is a no-rainbow $r$-coloring of $H$ or $|e \cap F'| \neq r-1$ for each edge $e \in E$.

%Procedure \localSearch takes as its arguments an $r$-uniform hypergraph $H = (X,E)$, a candidate pair $(c,F)$ for $H$, and an integer $g \ge 0$. We argue that \localSearch returns $1$ if and only if there exists a $(c,F)$ is within $g$ of a no-rainbow coloring.

\begin{lemma}\label{lem:localSearchAndFreeze}
Let $H$ be an $r$-uniform hypergraph, $(c,F)$ be a candidate pair for $H$ and $g$ be a positive integer.  Then, \localSearch{$H, (c,F),g$} returns $1$ if $(c,F)$ is within $g$ of a no-rainbow coloring, and $0$ otherwise.  The running time of \localSearch is $O^*((r-1)^g)$.
%Consider an execution of \localSearch together with all the recursive calls that are triggered by this. There exists a hyperedge with exactly $r-1$ frozen vertices. 
\end{lemma}

\begin{proof}
The first four \textbf{if} statements of \localSearch successively consider the base cases.  If $g = 0$ and $c$ induces a rainbow edge, no further color changes are allowed, so the answer must be $0$.  If $c$ induces a rainbow edge $e$ such that $e \subseteq F$, then we are not allowed to change any of the colors in $e$, so the answer must again be $0$.  If $c$ is a no-rainbow coloring, then we must obviously return $1$.
If $|e \cap F| \neq r-1$ for every edge $e$, then, by Lemma \ref{lem:freeze1}, $H$ has a no-rainbow $r$-coloring, so \localSearch returns $1$.  

For the recursion, suppose  $(c,F)$ is within $g$ of a no-rainbow coloring and that $e$ is any rainbow edge such that $|e \cap F| = r-1$.  Let $v$ be the single unfrozen node in $e - F$.  Then, there are only $r-1$ possible colors for $v$. By setting $v$ to one of these alternative colors and adding $v$ to $F$, we obtain a candidate pair that is within $g-1$ of a no-rainbow coloring.

The running time of \localSearch is $O^*((r-1)^g)$, since the depth of its recursion tree is $g$, each node has $0$ or $r-1$ children, and the work per node is polynomial.
\end{proof}

\begin{theorem}
Given an $n$-node $r$-uniform hypergraph $H$, algorithm \DetNRC returns $1$ if $H$ admits a no-rainbow, and $0$ otherwise.  The worst-case running time of \DetNRC is $O^*((r-1)^{(r-1)n/r})$.  In particular, the worst-case running time is $O^*(1.59^n)$ for $r=3$ and $O^*(2.28^n)$ for $r=4$.
\end{theorem}

\begin{proof}
Suppose $H$ has a no-rainbow $r$-coloring $c^*$.  Then, by Lemma \ref{lem:initial}, one of the initial candidate pairs is within $(r-1)n/r$ of $c^*$.  Since \localSearch tries all possible initial candidates, \DetNRC is correct.

By Lemma \ref{lem:localSearchAndFreeze}, since \localSearch is invoked with search radius $g = \frac{(r-1) n}{n}$, each call to \localSearch $O^*((r-1)^{(r-1)n/r})$ time.  Since \DetNRC calls \localSearch a polynomial number of times, the claimed running time follows.
\end{proof}

\section{Randomized Algorithm\label{sec:lSFreezeRand}}

Our randomized algorithm for the no-rainbow $r$-coloring problem (Algorithm \ref{alg:nonRainbowRandFreeze}) consists of a procedure \RandNRC whose outer \textbf{for} loop conducts a series of trials. Each trial invokes a randomized procedure \randLocalSearch to search for a no-rainbow $r$-coloring starting from different candidate pairs.  
We will prove that, given an $n$-node $r$-uniform hypergraph $H$ and a number $\alpha > 1$, \RandNRC{$H,\alpha$} finds a no-rainbow $r$-coloring of $H$, if such a coloring exists, with probability at least $1 - \frac{1}{e^\alpha}$.  
%Thus, the higher the value of $\alpha$, the lower the probability that \RandNRC{$H,\alpha$} incorrectly reports that $H$ does not have a no-rainbow $r$-coloring, when indeed $H$ does.

\begin{algorithm}%[H]

\Fn(){\RandNRC{$H,\alpha$}}{
%\KwIn{$r$-uniform hypergraph $H = (X,E)$.}
%\KwOut{}
\SetAlgoLined
\SetNoFillComment
\DontPrintSemicolon
\For{$j=1$ \KwTo $\alpha \cdot (\frac{r}{2})^n$}{
	\ForEach{$F \subseteq X$ such that $|F| = r$}{
		Let $c$ be the coloring obtained by assigning distinct colors to the nodes of $F$  and assigning each node in $X \setminus F$ a color chosen uniformly at random from $[r]$ \;
		\lIf{\randLocalSearch{$H,c,F$}}{\Return $1$}
	}
}
\Return $0$
}

\medskip

\Fn(){\randLocalSearch{$H,c,F$}}{\SetAlgoLined
\SetNoFillComment
\DontPrintSemicolon
%\While{there exists a rainbow hyperedge $e$ such that $|e\cap F| = r-1$}{
\For{$i =1$ \KwTo $n-r$}{
	\lIf{$c$ induces a rainbow edge such that $e \subseteq F$}
		{\Return $0$}	
	\lIf{$c$ is a no-rainbow coloring}
		{\Return $1$}
	\lIf{$|e \cap F| \neq r-1$ for every edge $e$}{
	%{assign color $c_1$ to all the unfrozen nodes\;
		\Return $1$}

	Pick an arbitrary rainbow hyperedge $e$ such that $|e\cap F| = r-1$\;
	Let $v$ be the node of $e$ such that $v \not\in F$ \;
	Pick a color $j \in [r] \setminus \{c(v)\}$ uniformly at random \;
	$c(v) = j$ \;
	$F = F \cup \{v\}$
}
\Return $0$
}
\caption{A randomized algorithm for no-rainbow $r$-coloring.}
\label{alg:nonRainbowRandFreeze}
\end{algorithm}

%\begin{lemma}\label{lem:prob}
% Suppose we have a randomized algorithm with probability $p > 0$ of success. Then $O(\frac{1}{p} \log n)$ repetitions of the algorithm are sufficient so that we have probability greater than or equal to $1-\frac{1}{n^{10}}$ that at least one of the runs was successful.
% \end{lemma}

%\begin{proof}
%We use the fact that $1-p \geq e^{-p}$ for any real number $p$. Now, suppose we repeat the algorithm $t$ times. The probability that all repetitions fail equals $(1-p)^t\leq e^{-pt}$. If we plug in $t= \frac{10}{p} \ln n$ we get $n^{-10}$. So, the probability that at least one of the repetitions succeeds is at least $1-n^{-10}$
%\end{proof}

We begin by bounding the probability that \randLocalSearch finds a no-rainbow $r$-coloring, provided we start it at the right point.

\begin{lemma}\label{lem:localSearch}
Let $H = (X,E)$ be an $n$-node $r$-uniform hypergraph. Suppose $(c,F)$ is a candidate pair such that $|F| = 4$ and for each node $v \in X \setminus F$, $c(v)$ is chosen at random. If $H$ has a no-rainbow $r$-coloring that agrees with $c$ on $F$, then, with probability at least $\left (\frac{2}{r} \right )^n$, \randLocalSearch{$H,c,F$} will find a no-rainbow $r$-coloring for $H$. 
\end{lemma}

\begin{proof}
Suppose $H$ has a no-rainbow $r$-coloring that agrees with $c$ on $F$.  \randLocalSearch starts from $(c,F)$ and at each step determines that either (1) it has reached a dead end, because there exists a rainbow edge all whose nodes are frozen, (2) $c$ is a no-rainbow $r$-coloring, (3) $(c,F)$ satisfies the conditions of Lemma \ref{lem:freeze1} and, therefore $H$ admits a no-rainbow $r$-coloring or (4) there is a rainbow edge,  say $e$, with exactly one unfrozen node.  In case (1), we say that \randLocalSearch \emph{fails}, while in cases (2) and (3) it \emph{succeeds}.  In case (4)  \randLocalSearch, recolors the unfrozen node, say $v$, in $e$ with a randomly chosen color, adds $v$ to $F$ and continues the search. 

The probability that \randLocalSearch succeeds is the probability that it reaches cases (2) or (3). Let $(c^*,F^*)$ be a candidate pair such that $c^*$ agrees with $c$ on $F$ and $F \subseteq F^*$ such that either $c^*$ is a rainbow coloring or $(c^*,F^*)$ satisfies the conditions of Lemma \ref{lem:freeze1}.  By assumption, $(c^*, F^*)$ must exist.  If at some iteration of \randLocalSearch we have $(c,F) = (c^*, F^*)$, then \randLocalSearch succeeds.  

Let $A_k$ be the event that the initial assignment $c$ disagrees with $c^*$ on exactly $k$ vertices. By assumption, $c$ agrees with $c^*$ on at least $r$ nodes. Thus, letting $m = n-r$, we have 
$$\Pr[A_k] = {m \choose k}(r-1)^k r^{-m}.$$
The probability $p$ that \randLocalSearch succeeds is at least the probability that each iteration reduces the distance between $c$ and $c^*$ by 1, until $c = c^*$. The probability of getting closer to $c^*$ at each step is $\frac{1}{r-1}$.
Thus, we have
\begin{equation*}
    p \geq \sum_{k=0}^m {m\choose k}(r-1)^k r^{-m}\left(\frac{1}{r-1}\right)^k 
       = r^{-m} \sum_{k=0}^m {m\choose k}
       =  \left(\frac{2}{r}\right)^m
       \geq \left(\frac{2}{r}\right)^n,
\end{equation*}
as claimed.
%\begin{align*}
%    p & \geq \sum_{k=0}^n {n\choose k}(r-1)^k r^{-n}\left(\frac{1}{r-1}\right)^k \\
%     & = r^{-n} \sum_{k=0}^n {n\choose k}\\
%     & = \left(\frac{2}{r}\right)^n
%\end{align*}
\end{proof}

\begin{theorem}
Given an $r$-uniform hypergraph $H = (X,E)$ and a number $\alpha > 1$, \RandNRC finds a no-rainbow $r$-coloring of $H$, if one exists, with probability at least $1- \frac{1}{e^\alpha}$.  The worst-case running time of  \RandNRC is $O^*((\frac{r}{2})^n)$.
In particular, for $r = 3$ and $r =4$, the worst-case running times are $O^*(1.5^n)$ and $O^*(2^n)$, respectively. 
\end{theorem}

\begin{proof}
If $H$ has a no-rainbow $r$-coloring, there must exist a set $F \subseteq X$ with $|F| = r$ such that the nodes in $F$ have $r$ different colors. Thus, by Lemma \ref{lem:localSearch}, at least one of the calls to \randLocalSearch has probability at least $\left (\frac{2}{r} \right )^n$ of success.  It follows from standard results (see, e.g., \cite[p.\ 151]{SchoeningToran2013}) that by conducting $\alpha \cdot (\frac{r}{2})^n$ trials of \RandNRC's outer \textbf{for} loop, the probability of success is at least $1- \frac{1}{e^\alpha}$.  The time bound follows by noting that  the running time of  \randLocalSearch is polynomial and, for each $j$, the number of executions of \RandNRC's inner \textbf{for} loop is polynomial for fixed $r$.
\end{proof}

\bibliographystyle{plain}
\bibliography{../Bibliographies/bibquartets,../Bibliographies/mybib,../Bibliographies/phylogenies}%,propBib}

\end{document}